\documentclass[runningheads]{llncs}

\usepackage{bbding}

\usepackage{graphicx}
\usepackage{amsfonts}
\usepackage{amssymb}
\usepackage[lined,boxed,commentsnumbered,ruled,vlined,noend,linesnumbered,boxed]{algorithm2e}

\usepackage{tcolorbox}
\usepackage{mathtools}

\usepackage{csquotes}
\usepackage{float}
\usepackage{framed} 

\newcommand{\np}{{NP}}

\newcommand{\eset}{{\it EC3Set}}
\newcommand{\stpg}{{\it STPG}}
\newcommand{\hpath}{{\it Ham-Path}}

\newcommand{\bcs}{{\textit{BCS}}}
\newcommand{\bcp}{{\textit{BCP}}}

\newcommand{\mnwcs}{{\textit{MNWCS}}}

\newcommand{\bcst}{{\textit{BCST}}}

\newcommand{\colb}[1]{{\textit{#1}}}

\usepackage{hyperref}

\begin{document}
	\title{The Balanced Connected Subgraph Problem}
	%
	%
	\author{Sujoy Bhore\inst{1,}\thanks{The author is partially supported by the Lynn and William Frankel Center for Computer Science, Ben-Gurion University of the Negev, Israel.}\and
		Sourav Chakraborty\inst{2} \and
		Satyabrata Jana\inst {2}\and \\
		Joseph S. B. Mitchell\inst{3} \and
		Supantha Pandit\inst{3,}\thanks{The author is partially supported by the Indo-US Science \& Technology Forum (IUSSTF) under the SERB Indo-US Postdoctoral Fellowship scheme with grant number 2017/94, Department of Science and Technology, Government of India.}\and
		Sasanka Roy\inst{2}}
	\authorrunning{S. Bhore et al.}
	%
	\institute{Ben-Gurion University, Beer-Sheva, Israel \and Indian Statistical Institute, Kolkata, India \and
		Stony Brook University, Stony Brook, NY, USA
		\\ \email{\{sujoy.bhore,chakraborty.sourav,satyamtma,\\joseph.s.b.mitchell,pantha.pandit,sasanka.ro\}@gmail.com}}

	\maketitle              
	\begin{abstract}
		The problem of computing induced subgraphs that satisfy some specified restrictions arises in various applications of graph algorithms and has been well studied.
		In this paper, we consider the following \emph{Balanced Connected Subgraph (shortly, BCS)} problem. 
		The input is a graph $G=(V,E)$, with each vertex in the set $V$ having an assigned color, ``red'' or ``blue''. 
		We seek a maximum-cardinality subset $V'\subseteq V$ of vertices that is {\em color-balanced} (having exactly $|V'|/2$ red nodes and $|V'|/2$ blue nodes), 
		such that the subgraph induced by the vertex set $V'$ in $G$ is connected.
		We show that the BCS problem is \np-hard, even for bipartite graphs $G$ (with red/blue color assignment not necessarily being a proper 2-coloring). 
		Further, we consider this problem for various classes of the input graph $G$, including, e.g., planar graphs, chordal graphs, trees, split graphs, bipartite graphs with a proper red/blue $2$-coloring, and graphs with diameter $2$. 
		For each of these classes either we prove \np-hardness or design a polynomial time algorithm. 
		
		\keywords{Balanced connected subgraph \and Trees \and  Split graphs \and  Chordal graphs \and  Planar graphs \and Bipartite graphs \and  \np-hard \and Color-balanced.}
	\end{abstract}
	\section{Introduction}
	Several problems in graph theory and combinatorial optimization involve determining if a given graph $G$ has a subgraph with certain properties.
	Examples include seeking paths, cycles, trees, dominating sets, cliques, vertex covers, matching, independent sets, bipartite subgraphs, etc. 
	Related optimization problems include finding a maximum clique, a maximum (connected) vertex cover, a maximum independent set, a minimum (connected) dominating set, etc.  
	These well-studied problems have significant theoretical interest and many practical applications. 
	
	In this paper, we consider the problem in which we are given a simple connected graph $G=(V,E)$ whose vertex set $V$ has each node being ``red'' or ``blue'' 
	(note, the color assignment might not be a proper $2$-coloring of the vertices, i.e., we allow nodes of the same color to be adjacent in $G$).
	We seek a maximum-cardinality subset $V'\subseteq V$ of the nodes such that $V'$ is {\em color-balanced}, i.e. having same number of red and blue nodes in $V'$, and such that the induced subgraph $H$ by $V'$ in $G$ is connected.  We refer to this problem as the \colb{Balanced Connected Subgraph (\bcs)} problem:
	
	\begin{framed}
		\noindent {\bf \textit{{\emph{Balanced Connected Subgraph (\textit{BCS}) Problem}}}}\\
		{\bf Input:} A graph $G=(V,E)$, with node set $V=V_R\cup V_B$ partitioned into red nodes ($V_R$) and blue nodes ($V_B$).\\
		{\bf Goal:} Find a maximum-cardinality color-balanced subset $V'\subseteq V$ that induces a connected subgraph $H$.
	\end{framed}

	Notice that, the \bcs~problem is a special case of the \colb{Maximum Node Weight Connected Subgraph (\mnwcs)} problem \cite{Johnson1985}. 
	In the \mnwcs~problem, we are given a connected graph $G(V,E)$, with a (possibly negative) integer weight $w(v)$ associated with each node $v\in V$, and an integer bound $B$; the objective is to decide whether there exists a subset $V'\subseteq V$ such that the subgraph induced by $V'$ is connected and the total weight of the vertices in $V'$ is at least $B$. In the \mnwcs~problem, if we assign the weight of each vertex is either $+1$ (red) or $-1$ (blue), then deciding whether there exist a $V'\subseteq V$ such that $|V'| \geq k$, subgraph induced by $V'$ in $G$ is connected and total of vertices in $V'$ is exactly zero 
	is equivalent as the \bcs~problem. The \mnwcs~problem along with its variations have numerous practical application in various fields. This includes designing fiber-optic networks \cite{Felix1998}, oil-drilling \cite{Hochbaum1994}, systems biology \cite{Backes2012,Dittrich2008,Yamamoto2009}, wildlife corridor design \cite{Dilkina2010}, computer vision \cite{Chen2017,Chen2012}, forest planning \cite{Carvajal2013}, and many more (see \cite{El-Kebir2014} and the references therein). 
	Some of these applications are best suited to the \bcs~problem.
	

	\subsection{Related work}
	
	The bichromatic inputs, often referred in the literature as red-blue input, has appeared extensively in numerous 
	problems. For bipartite trees, see \cite{abellanas1999bipartite}. In \cite{biniaz2014bottleneck,dumitrescu2001matching,dumitrescu2002partitioning}
	colored points have been considered in the context of matching and partitioning problem. For a detailed survey on geometric problems with red-blue points; see \cite{kaneko2003discrete}.
	In \cite{aichholzer2015balanced}, Aichholzer et al. considered the balanced island problem and devised polynomial algorithms for points considered on plane. From combinatorial side, Balanchandran et al. \cite{balachandran2017system} studied the problem of unbiased representatives in a set of bicolorings.  In this paper, they have mentioned the usefulness of the unbiased representatives in drug testing. While the drugs are tested over a large population, the effectiveness of a new drug is measured under various attributes e.g., weight, height, age etc. One would require to sample representative in certain \emph{balanced} manner. Kaneko et al. \cite{kaneko2017balancing} considered the problem of balancing the colored points on the line. Subsequently, Bereg et al. \cite{bereg2015balanced} studied the balanced partitions of $3$-colored geometric sets on the plane. 
	
	On the other hand, finding a certain type of subgraph in a graph is considered to be a fundamental algorithmic question. In \cite{feige2001dense}, Feige et al. studied the dense $k$-subgraph problem where given a graph $G$ and a parameter $k$, the goal is to find a set of $k$ vertices with maximum average degree in the subgraph induced by this set. From parameterized algorithms side, Crowston et al. \cite{crowston2013maximum} considered the balanced subgraph problem. Kierstead et al. \cite{kierstead1992colorful} studied the problem of finding colorful induced subgraph in a properly colored graph. This led us to study the balanced connected subgraph problem on graphs.
	In \cite{derhy2009finding}, Derhy and Picouleau considered the problem of finding induced trees on both weighted and unweighted graphs and obtained hardness and algorithmic results. 
	They have studied some particular classes of graphs like the bipartite graphs or the triangle-free graphs. Moreover, they have considered the case where the number of prescribed vertices is bounded.

	\subsection{Our contributions} 
	In this paper, we consider the balanced connected subgraph problem on various graph families and present several hardness and algorithmic results. 
	
	On the hardness side, in Section \ref{np-hardness}, we prove that the \bcs~problem is \np-hard on general graphs, even for planar graphs, bipartite graphs (with a general red/blue color assignment, not necessarily a proper 2-coloring), and chordal graphs. 
	Furthermore, we show that the existence of a balanced connected subgraph containing a specific vertex is \np-complete. In addition to that, we prove that finding the maximum balanced path in a graph is \np-hard. 
	
	On the algorithmic side, in Section \ref{algorithmicresults}, we devise polynomial-time algorithms 
	for trees (in $O(n^4)$ time), split graphs (in $O(n^2)$ time), bipartite graphs with a proper 2-coloring (in $O(n^2)$ time), and graphs with diameter~$2$ (in $O(n^2)$ time). Here $n$ is the number of vertices in the input graphs.



	
	\section{Hardness results}\label{np-hardness}

	\subsection{\bcs~problem}\label{NPBCP}
	In this section we prove that the \bcs~problem is \np-hard for bipartite graph with a general red/blue color assignment, not necessarily a proper 2-coloring. We give a reduction from the \colb{Exact-Cover-by-3-Sets (\eset)} problem \cite{Garey1979}. 
	In this \eset~problem, we are given a set $U$ with $3k$ elements and a collection $S$ of $m$ subsets of $U$ such that each $s_i\in S$ contains exactly $3$ elements. 
	The objective is to find an exact cover for $U$ (if exists), i.e., a sub-collection $S'\subseteq S$ such that every element of $U$ occurs in exactly one member of $S'$. During the reduction, we generate an instance $G=(R\cup B,E)$ of \bcs~problem from an instance $X(S,U)$ of the \eset~problem as follows:	
	
	\noindent {\bf Reduction:} For each set $s_i\in S$, we take a blue vertex $s_i \in B$. For each element $u_j\in U$, we take a red vertex $u_j \in R$. Now consider a set $s_i \in S$ which contains three elements $u_\alpha, u_\beta$, and $u_\gamma$, then we add 3 edges $(s_i,u_\alpha)$, $(s_i,u_\beta)$, and $(s_i,u_\gamma)$ in $E$. Additionally, we consider a path of $5k$ blue vertices starting and ending with vertices $b_1$ and $b_{5k}$ respectively. Similarly, we consider a path of $3k$ red vertices starting and ending with vertices $r_1$ and $r_{3k}$ respectively. We connect these two paths by joining the vertices $r_{3k}$ and $b_1$ by an edge. Finally, we connect each vertices $s_i$ with $b_{5k}$ by edges. This completes the construction. See Figure \ref{fig:my_label} for the complete construction.
	
	\begin{figure}[ht]
		\centering
		\includegraphics[scale=.6]{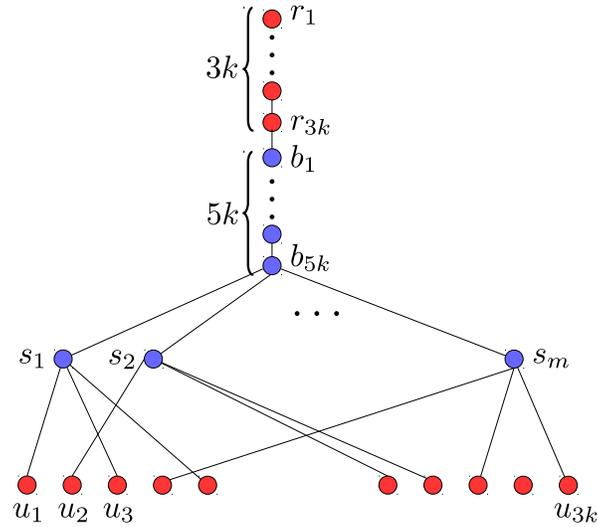}
		\caption{Construction of the instance $G$ of the \bcs~problem.}\label{fig:my_label}
	\end{figure}
	
	Clearly, the number of vertices and edges in $G$ are polynomial in terms of number of elements and sets in $X$. Hence the construction can be done in polynomial time. We now prove the following theorem.
	
	\begin{lemma}\label{lem-bcs-np-hard}
		The instance $X$ of the \eset~problem has a solution if and only if the instance $G$ of the \bcs~problem has a connected balanced subgraph $T$ with $12k$ vertices ($6k$ red and $6k$ blue).
	\end{lemma} 
	
	\begin{proof}
		Assume that \eset~problem has a solution. Let $S^*$ be an optimal solution in it. We choose the corresponding vertices of $S^*$ in $T$. Since this solution covers all $u_j$'s. So we select all $u_j$'s in $T$. Finally we select all the $5k$ blue and $3k$ red vertices in $T$, resulting in a total of $6k$ red and $6k$ blue vertices.
		
		On the other hand, assume that there is a balanced tree $T$ in $G$ with $6k$ vertices of each color. The solution must pick the $5k$ blue vertices $b_1, \ldots, b_{5k}$. Otherwise, it exclude the $3k$ red vertices $r_1, \ldots,r_{3k}$, and reducing the size of the solution. Since the graph $G$ has at most $6k$ red vertices, at most $k$ vertices can be picked from the set $s_1, \ldots,s_m$ and need to cover all the $3k$ red vertices corresponding to $u_j$ for $1\leq j\leq 3k$. Hence, this $k$ sets give an exact cover. \qed
	\end{proof}
	
	It is easy to see that the graph we constructed from the \colb{Exact-Cover-by-3-Sets (\eset)} problem in Figure \ref{fig:my_label} is indeed a bipartite graph. Hence we have the following theorem.
	
	\begin{theorem}\label{thm-bcs-np-hard-bipartite} 
		\bcs~problem is \np-hard for bipartite graphs.
	\end{theorem}

	\subsection{\np-hardness: \bcs~problem over special classes of graphs}
	In this section, we show that the \bcs~problem is \np-hard even if we restrict the graph classes to chordal, or planar graphs.
	\subsubsection{Chordal graphs:}
	We prove that the \bcs~problem is \np-hard where the input graph is a chordal graph. The hardness construction is similar to the construction in Section~\ref{NPBCP}; we modify the construction so that the graph is chordal. In particular, we add edges between $s_{i}$ and $s_{j}$ for each $i \neq j, 1 \leq i,j \leq m$. For this modified graph, it is easy to see that a lemma identical to Lemma~\ref{lem-bcs-np-hard} holds. Hence, we conclude that the \bcs~ problem is \np-hard for chordal graphs.
	
	\subsubsection{Planar graphs:}
	In this section we prove that \bcs~problem is \np-hard for planar graphs. We give a reduction from the Steiner tree problem in planar graphs (\stpg) \cite{Garey1979}. In this problem, we are given a planar graph $G = (V,E)$, a subset $X \subseteq V$, and a positive integer $k \in \mathbb{N}$. The objective is to find a tree $T =(V',E')$ with at most $k$ edges such that $ X \subseteq V'$.

	\vspace{.2cm}
	\noindent {\it Reduction:} We generate an instance $H=(R \cup B,E(H))$ for the \bcs~problem from an instance $G = (V,E)$ of the \stpg~problem. We color all the vertices in $G$ as blue. 
	We now create red color vertices and connect to these vertices. For each vertex $u_i \in X$, we add a vertex $u'_i$ in $H$ whose color is red add connect $u'_i$ to $u_i$ via an edge. Additionally, we take a set $Z$ of $(k+1-|X|)$ red vertices in $H$ and the edges $(z_{j},u'_1)$ into $E(H)$, for each $z_{j} \in Z$. Hence we have, $B =V$, and $R= Z \cup \{u'_i\} | \ 1 \leq i \leq |X|\} $. Note that $|R|<|B|$ and $ |R|=(k+1) $. This completes the construction. For an illustration see Figure \ref{fig:planar}.
	\begin{figure}[ht]
		\centering
		\includegraphics[scale=.6]{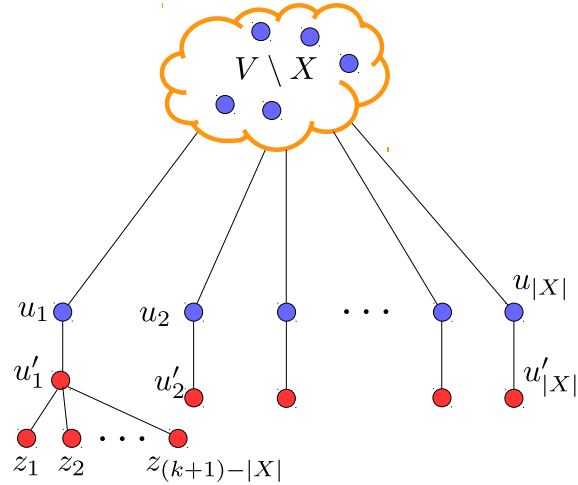}
		\caption{Schematic construction for planar graphs.}
		\label{fig:planar}
	\end{figure}
	Clearly the number of vertices and edges in $H$ are polynomial in terms of vertices in $G$. Hence the construction can be done in polynomial time. We now prove the following theorem.
	
	\begin{theorem}\label{theoBCS}
		\stpg~has a solution if and only if $H$ of the \bcs~problem has a balanced connected subgraph with $(k+1)$ vertices of each color.
	\end{theorem}
	
	\begin{proof}
		Assume that \stpg~has a solution. Let $T=(V',E')$ be the resulting Steiner tree which contains at most $k$ edges and $X\subseteq V'$. If $ |V'|=(k+1) $ then the subgraph of $H$ induced by $(V' \cup R)$ is connected and balanced with $(k+1)$ vertices of each color. If $ |V'| < (k+1) $ then we take a set $ Y $ of $((k+1)-|V'|)$ many vertices from $ V $ such that the subgraph of $G$ induced by $(V' \cup Y)$ is connected. Clearly $ |V'|=(k+1) $. Now the subgraph of $H$ induced by $(V' \cup Y \cup R)$ is connected and balanced with $(k+1)$ vertices of each red and blue color.  
		
		On the other hand, assume that there is a balanced connected subgraph $H'$ of $H$ with $(k+1)$ vertices of each color. Note that, except vertex $u'_1$, in $H$ all the red vertices are of degree $1$ and connected to blue vertices. Let $G'$ be the subgraph of $G$ induced by all blue vertices in $H'$. Since $H$ is connected and there is no edge between any two red vertices,  $G'$ is connected. Since $G'$ contains $(k+1)$ vertices, any spanning tree $T$ of $H'$ contains $k$ edges. So $T$ is a solution of \stpg~problem. \qed
	\end{proof}
	Hence we have the following theorem.
	
	\begin{theorem}\label{thm-bcs-np-hard} 
		\bcs~problem is \np-hard for planar graphs.
	\end{theorem}

	\subsection{NP-completeness for \bcs~problem for a specific vertex.}

	In this section we prove that the existence of a balanced subgraph containing a specific vertex is \np-complete. We call this problem the \bcs-existence problem. The reduction is similar to the reduction used in showing the \np-hardness of the \bcs~problem; we also use here a reduction from the \eset~problem (see Section \ref{NPBCP} for the definition).

	\noindent {\bf Reduction:} Assume that we are given a \eset~problem instance $X=(U,S)$, where set $U$ contains $3k$ elements and a collection $S$ of $m$ subsets of $U$ such that each $s_i\in S$ contains exactly $3$ elements.  We generate an instance $G(R,B,E)$ of the \bcs-existence problem from $X$ as follows. The red vertices $R$ are the elements $u_j\in U$; i.e., $R=U$. The blue vertices $B$ are the $3$-element sets $s_i\in S$; i.e., $B=S$. 
	For each blue vertex $s_i=\{u_\alpha,u_\beta,u_\gamma\} \in S=B$, we add the 3 edges $(s_i,u_\alpha)$, $(s_i,u_\beta)$, and $(s_i,u_\gamma)$ to the set $E$ of edges of $G$.
	%
	%
	We instantiate an additional set of $2k$ blue vertices, $\{b_1,\ldots,b_{2k}\}$, and add edges to $E$ to link them into a path $(b_1,b_2,\ldots,b_{2k})$. Finally, we add an edge from $b_{2k}$ to each of the blue vertices $s_i$. Refer to Figure \ref{fig:vertex}. 
	
	\begin{figure}[ht]
		\centering
		\includegraphics[scale=.6]{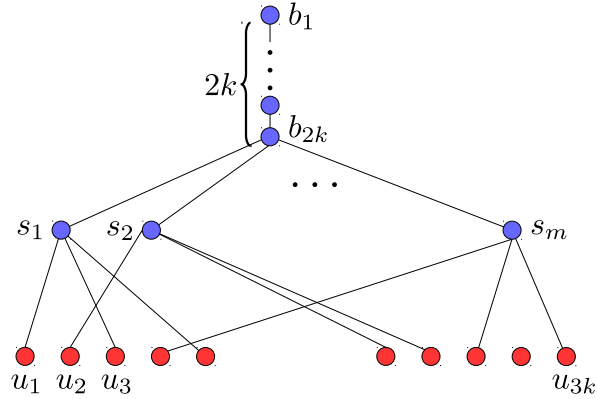}
		\caption{Construction of the instance $G$ of the \bcs~problem containing $ b_{1}$.}\label{fig:vertex}
	\end{figure}
	Clearly, the number of vertices and edges in $G$ are polynomial in terms of number of elements and sets in the size of the \eset~problem instance $X$, and the construction can be done in polynomial time. We now prove the following lemma.
	
	\begin{lemma}\label{lemma-bcs-np-hard}
		The instance $X$ of the \eset~problem has a solution iff the instance $G$ of the corresponding \bcs~existence problem has a balanced subgraph $T$ containing the vertex $b_1$.
	\end{lemma} 
	
	\begin{proof}
		Assume that the \eset~problem has a solution, and let $S^*$ be the collection of $k=|S^*|$ sets of $S$ in the solution. 
		%
		%
		Then, we obtain a balanced subgraph $T$ that contains $b_1$ as follows: $T$ is the induced subgraph of the $3k$ red vertices $U$, together with the $k$ blue vertices $S^*$ and the $2k$ blue vertices $b_1,\ldots,b_{2k}$. Note that $T$ is balanced and connected and contains $b_1$.
		
		Conversely, assume there is a balanced connected subgraph $T$ containing $b_{1}$. Let $t$ be the number of (blue) vertices of $S$ within $T$. First, note that $t\leq k$. (Since $T$ is balanced and contains at most $3k$ red vertices, it must contain at most $3k$ blue vertices, $2k$ of which must be $\{b_1,\ldots,b_{2k}\}$, in order that $T$ is connected.) 
		%
		%
		Next, we claim that, in fact, $t\geq k$. To see this, note that each of the $t$ blue vertices of $T$ that corresponds to a set in $S$ is connected by edges to $3$ red vertices; thus, $T$ has at most $3t$ red vertices. Now, $T$ has $2k+t$ blue vertices (since it has $t$ vertices other than the path $(b_1,\ldots,b_{2k})$), and $T$ is balanced; thus, $T$ has exactly  $2k+t$ red vertices, and we conclude that $2k+t\leq 3t$, implying $k\leq t$, as claimed. 
		%
		%
		Therefore, we need to select exactly $k$ blue vertices corresponding to the sets $S$, and these vertices connect to all $3k$ of the red vertices. The $k$ sets corresponding to these $k$ blue vertices is a solution for the \eset~problem. \qed
	\end{proof}
	
	It is easy to prove that the \bcs~existence problem is in \np. Hence, we have the following theorem.
	
	\begin{theorem}
		It is \np-complete to decide if there exists a connected balanced subgraph that contains a specific vertex.
	\end{theorem}
	
	\subsection{\np-hardness: balanced connected path problem}
	In this section we consider the balanced connected path (\bcp) Problem and prove that it is \np-hard. In this problem instead of finding a balanced connected subgraph, our goal is to find a balanced path with a maximum cardinality of vertices. To prove the \bcp~problem is \np-hard we give a polynomial time reduction from the \colb{Hamiltonian Path (Ham-Path)} problem which is known to be \np-complete \cite{Garey1979}. In this problem, we are given an undirected graph $Q$, and the goal is to find a Hamiltonian path in $Q$ i.e., a path which visits every vertex in $Q$ exactly once. In the reduction we generate an instance $G$ of the \bcp~problem from an instance $Q$ of the \hpath~problem as follows:
	
	\noindent  {\bf Reduction:} We make a new graph $Q'$ from $Q$. Let us assume that the graph $Q$ contains $m$ vertices. If $m$ is even then $Q'=Q$. If $m$ is odd, then we add a dummy vertex $u$ in $Q$ and connect to every other vertices in $Q$ by edges with $u$. The resulting graph is our desired $Q'$. It is easy to observe that, $Q$ has a Hamiltonian path if and only if $Q'$ has a Hamiltonian path.
	
	Now we have a \hpath~instance $Q'$ with even number of vertices, say $n$. We arbitrary choose any $n/2$ vertices in $Q'$ and color them red and color the remaining $n/2$ vertices blue. Let $G$ be the colored graph.
	
	This completes the construction. Clearly, this can be done in polynomial time. We now have the following lemma.
	
	\begin{lemma} $Q'$ has a Hamiltonian path $T$ if and only if $G$ has a balanced path $P$ with exactly $n$ vertices.
	\end{lemma}
	
	\begin{proof} Assume that $Q'$ has a Hamiltonian path $T$. This implies that, $T$ visits every vertex in $Q'$. Since by the construction there are exactly half of the vertices in $G$ is red and remaining are blue, the same path $T$ is balanced with $n/2$ vertices of each color.
		
		On the other hand, assume that there is a balanced path $P$ in $G$ with  exactly $n/2$ vertices of each color. Since, $G$ has a total of $n$ vertices, the path $P$ visits every vertex in $G$. Hence, $P$ is a Hamiltonian path. \qed
	\end{proof}
	
	Therefore, we have the following theorem.
	
	\begin{theorem}
		\bcp~problem is \np-hard for general graph.
	\end{theorem}
	
	\section{Algorithmic results}\label{algorithmicresults}
	
	
	In this section, we consider several graph families and devise polynomial time algorithms for the 
	\bcs~problem. 
	Notice that, if the graph is a path or cycle, the optimal solution is just a path. Hence, one can do brute-force search to obtain the maximum balanced path.  
	In case of a complete graph $K_n$, we output a sub-graph $H$ of $K_{n}$ induced by $V$, 
	where $|V|= 2|B|$, $B \subset V$, and $B$ is the set of all blue vertices in $K_{n}$ (assuming that, the number of blue vertices is at most the number of red vertices in $K_{n}$). 
	Clearly, $H$ is the maximum-cardinality balanced sub-graph in $K_{n}$.     
	We consider trees, split graphs, bipartie graphs (properly colored),
	graphs of diameter~$2$, and present polynomial algorithms for each of them. 
	
	
	\subsection{Trees}
	In this section we give a polynomial time algorithm for the \bcs~problem where the input graph is a tree. We first consider the following problem.
	
	\noindent {\bf Problem 1:} Given a tree $T=(V,E)$, and a root $t \in V$ where $V=V_R\cup V_B$. The vertices in $V_R$ and $V_B$ are colored red and blue,respectively. 
	The objective is to find maximum balanced tree with root $t$.
	
	We now design an algorithm to solve this problem. Let $v$ be a vertex in $G$. We associate a set $P_v$ of \colb{pairs} of the form $(r,b)$ to $v$, where $r$ is the count of red vertices and $b$ is the count of blue vertices.  A single pair  $(r,b)$ associated with vertex $v$ indicates that there is a  subtree rooted at $v$ having $r$ red and $b$ blue vertices. Note that $r$ may not be equal to $b$. Now for any $k$ pairs, the sum is also a pair which is defined as the element-wise sum of these $k$ pairs. Let $A_1,A_2, \ldots,A_k$ be $k$ sets. The Minkowski sum $^M\sum_{i=1}^k A_i$ denotes the set of sums of $k$ elements one from each set $A_i$ i.e., $^M\sum_{i=1}^k A_i =  A_1 \oplus A_2 \oplus  \ldots \oplus A_k$. We use $\oplus$ to denote Minkowski sum between sets. For example, for the Minkowski sum of the sets $A$ and $B$, we write $A \oplus B$ and it means $A \oplus B=\{a+b \colon a\in A, b\in B\}$.
	
	Now we are ready to describe the algorithm to solve Problem 1. In Algorithm 1, we describe how to get maximum balanced subtree with root $t$ for a  tree $T$ rooted at $t$. 
	
	\IncMargin{1em}
	\begin{algorithm}[H]
		
		\SetKwData{Left}{left}\SetKwData{This}{this}\SetKwData{Up}{up}
		\SetKwFunction{Union}{Union}\SetKwFunction{FindCompress}{FindCompress}
		\SetKwInOut{Input}{Input}\SetKwInOut{Output}{Output}
		\Input{$(i)$ A rooted tree  $T=(B \cup R,E)$ with root $t$.\\ $(ii)$ $B$ and $R$ are colored blue and red respectively.}
		\Output{A set of pairs at each node in $T$.}
		\BlankLine
		
		\If{$v$ is a leaf with red color}{
			{\ $P_{v}=\{(0,0),(1,0)\}$;}
		}
		\If{$v$ is a leaf with blue color}{
			{\ $P_{v}=\{(0,0),(0,1)\}$;}
		}
		\If{$v$ be a vertex with red color and $v$ has $k$ children $u_{1}, u_{2},...,u_{k}$ in $T$ with root at $r$,}{
			{\ $P_{v} = \{(0,0)\} \cup  \{ ^M\sum_{i=1}^k P_{u_{i}} \oplus \{(1,0)\} \}$;}
		} 
		\If{$v$ be a vertex with blue color and $v$ has $k$ children $u_{1}, u_{2},...,u_{k}$ in $T$ with root at $r$,}{
			{\ $P_{v} = \{(0,0)\} \cup  \{ ^M\sum_{i=1}^k P_{u_{i}} \oplus \{(0,1)\} \}$}\tcp*[r]{\ \colb{ \color{blue} Here $\oplus$ denotes Minkowski Set sum.}}
		}
		\Return {$P_{t}$}
		\caption{ Construct red-blue pair-sets in a rooted tree.}\label{algo_1}
	\end{algorithm}\DecMargin{1em}

	In Algorithm \ref{algo_1}  we compute a finite set $P_{t}$ of pairs $\{(r,b)\}$ at the root $t$ in $T$. To do so, we recursively calculate the set of pairs from leaf to the root. For an internal vertex $v$, the set $P_v$ is calculated as follows: let the color of $v$ is red and it has $k$ children $u_1,u_2,\ldots,u_k$. Then, $P_{v} = \{(0,0)\} \cup \{^M\sum_{i=1}^{k} P_{u_i} \oplus \{(1,0)\}\}$.

	We now prove the following lemma.
	\begin{lemma}
		Let $T$ be rooted tree with $t$ as a root. Then Algorithm \ref{algo_1} produces all possible balanced subtree rooted at $t$ in $O(n^6)$ time.
	\end{lemma}
	
	\begin{proof}
		Notice that in Algorithm \ref{algo_1}, at each node $v\in T$, we store a set $P_{v}$ of pairs $\{(r_{i},b_{i})\}$, where each $(r_{i},b_{i})$ indicates that there exists a subtree $T'$ with root $v$ such that number of red and blue vertices in $T'$ are $r_{i}$ and $b_{i}$, respectively. Note that $r_{i}$ may not be same as $b_{i}$. When we construct the set $P_v$, all the sets corresponding to its children are already calculated. Finally, in steps $6$ and $8$ of Algorithm \ref{algo_1} we calculate the set $P_v$ based on the color of $v$. Hence, when Algorithm \ref{algo_1} terminates, we get the set $P_t$ where $t$ is the root of $T$.
		
		Now we calculate the time taken by Algorithm \ref{algo_1}. Clearly, steps $2$ and $4$ take $O(1)$ time to construct the $p_v$ when $v$ is a leaf. Note that, the size of $P_v$, for an internal node $v$ is $O(n^2)$. Since there are at most $n$ blue and red vertices in the subtree rooted at $v$. If $ v $ has $k$ children then we have to take Minkowski sum of the sets corresponds to the children of $v$. To get the sum of two sets it takes $ O(n^4) $ time. As there are at most $ n $ children of node $ v $, so the time taken by steps 6 and 8 are $O(n^5)$. Finally, we traverse  the tree from bottom to the root. Hence, the total time taken by the algorithm is $O(n^6)$. \qed
	\end{proof}

	We can now improve the time complexity by slightly modifying the Algorithm \ref{algo_1}. For an internal vertex $v$, we actually don't need all the pairs to get the maximum balanced subtree. Suppose there are two pairs $(a,b)$ and $(c,d)$ in $P_{v}$, where $(b-a)=(d-c)$ and $a<c$. Then, instead of using the subtree with pair $(a,b)$, it is better to use the subtree with pair $(c,d)$, since it may help to construct a larger balance subtree. Therefore, in a set $P_{v}$ if there are $k$ pairs $\{(a_{i},b_{i}); 1 \leq i \leq k \}$ such that $(b_{i}-a_{i})=(b_{j}-a_{j})$ whenever $i \neq j$, $1 \leq i,j \leq k$. Then we remove the $(k-1)$ pairs and store only the pair which is largest among all these $k$ pairs. We say $(a_{m},b_{m})$ is largest when $a_{m} > a_{i}$ and $b_{m}> b_{i}$ for $1 \leq i \leq k, i \neq m$. So we reduce the size of $P_{v}$ for each vertex $v \in T$ from $ O(n^2) $ to $O(n)$.  Let $T(n)$ be the time to compute red-blue pairset for the root vertex $t$ in the tree $T$ with size $n$. If $r$ has $k$ children $u_{1}, u_{2},...,u_{k}$ with size $n_{1},n_{2},...,n_{k}$. Then the recurrence is $T(n)= T(n_{1})+T({n_{2})+...+T(n_{k})+O(\sum_{i=1}^{k-1} (n_{1}+ n_{2}+ \dots + n_{i})n_{i+1})}$. Now $ \sum_{i=1}^{k-1} (n_{1}+ n_{2}+ \dots n_{i})n_{i+1} \leq  \sum_{i=1}^{k-1} nn_{i+1} = n  \sum_{i=1}^{k-1} n_{i+1}  \leq n^2$. 
	which gives the solution that $T(n)= O(n^3)$. 
	
	Hence, we conclude the following lemma. 
	\begin{lemma}
		Let $T$ be rooted tree with $t$ as a root. We can produces all possible balanced subtree rooted at $t$ in $O(n^3)$ time and 
		$ O(n^2) $ space complexity.
	\end{lemma}
	
	\subsection*{Optimal solution for BCS problem in tree}
	If there are $n$ nodes in the tree $T$, then, for each node $v_{i}, 1 \leq i \leq n$, we consider $T$ to be a tree rooted at $v_{i}$. We then apply Algorithm $\ref{algo_1}$ to find maximum-cardinality balanced subtree rooted at $v_{i}$; let $T_i$ be the resulting balanced subtree, having $m_{i}$ vertices of each color. Then, to obtain an optimal solution for the \bcst~problem in $T$ we choose a balanced subtree that has $max \{m_i; 1 \leq i \leq n\}$ vertices of each color. Now we can state the following theorem.
	\begin{theorem}
		Let $T$ be a tree whose $n$ vertices are colored either red or blue. Then, in $O(n^4)$ time and $O(n^2)$ space, one can compute a maximum-cardinality balanced subtree of~$T$.
	\end{theorem}

	\subsection{Split graphs}
	A graph $G = (V, E)$ is defined to be a split graph if there is a
	partition of $V$ into two sets $S$ and $K$ such that $S$ is an independent set and $K$ is a complete graph. There is no restriction on edges between vertices of $S$ and $K$.
	Here we give a polynomial time algorithm for the \bcs~problem where the input graph $G=(V,E)$ is a split graph. Let $S$ and $K$ be two disjoint partition of $V$ where $S$ is an independent set and $K$ is a complete graph. Also, let $S_B$ and $S_R$ be the sets of blue and red vertices in $S$, respectively. Similarly, let $K_B$ and $K_R$ be the sets of blue and red vertices in $K$, respectively. We argue that there exists a balanced connected subgraph in $G$, 
	having $\min\{|S_B\cup K_B|,|S_R\cup K_R|\}$ vertices of each color.
	
	Note that if $|S_{B}\cup K_{B}|=|S_{R}\cup K_{R}|$ then $G$ itself is balanced. 
	Now, w.l.o.g., we can assume that $|S_{B}\cup K_{B}|<|S_{R}\cup K_{R}|$. 
	We will find a connected balanced subgraph $H$ of $G$, where the number of vertices in $H$ is exactly $2|S_{B}\cup K_{B}|$. To do so, we first modify the graph $G=(V,E)$ to a graph $G'=(V,E')$. Then, from $G'$, we will find the desired balanced subgraph with $|S_B\cup K_B|$ many vertices of each color. Moreover, this process is done in two steps. 
	
	\begin{description}
		\item[Step~1:] Construct $G'=(V,E')$ from $G=(V,E)$.\\
		For each $u \in S_{B}$, if $u$ is adjacent to at least a vertex $u'$ in $K_{R}$, then remove all adjacent edges with $u$ except the edge $(u,u')$. Similarly, for each $v \in S_{R}$, if $v$ is adjacent to at least a vertex $v'$ in $K_{B}$, then remove all adjacent edges with $v$ except the edge $(v,v')$.
		
		\vspace{.2cm}
		
		\item[Step~2:]
		Delete $|S_{R}\cup K_{R}|-|S_{B}\cup K_{B}|$ vertices from $G'$.\\
		Let $k= |S_{R}\cup K_{R}|-|S_{B}\cup K_{B}|$. Now we we have following cases.

		\begin{description}
			\item[Case 1:] $|S_{R}| \geq k$. We remove $k$ vertices from $S_{R}$ in $G'$. Clearly, after this modification, $G'$ is connected, and we get a balanced subgraph having $|S_B\cup K_B|$ vertices of each color.
			\item[Case 2:] $|S_{R}|< k$. 	
			Then we know, $|K_{R}| > |K_{B}\cup S_{B}|$.
			Let $S'_{B} \subseteq S_{B}$ be the set of vertices in $G'$ such that each vertex of $S'_{B}$ has exactly one neighbor in $K_R$. 
			Then, we take a set $ X\subset K_{R} $ with cardinality $|K_{B} \cup S_{B}|$ such that $X$ contains all adjacent vertices of $S'_{B}$.  
			Now we take the subgraph $H$ of $G'$ induced by 
			$(S_{B} \cup K_{B} \cup X )$. $H$ is optimal and balanced.
		\end{description}
		
	\end{description}

	\vspace{.2cm}
	
	\noindent {\bf Running time:} Step 1 takes $O(|E|)$ time to construct $G'$ from $G$. Now in step~2, both  Case~1 and Case~2 take $O(|V|)$ time to delete $|S_{R}\cup K_{R}|-|S_{B}\cup K_{B}|$ vertices from $G'$. Hence, the total time taken is $O(n^2)$, where $n$ is the number of vertices in $G$. We conclude in the following theorem.
	
	\begin{theorem}
		Given a split graph $G$ of $n$ vertices, with $r$ red and $b$ blue ($n=r+b$) vertices, then, in $O(n^2)$ time we can find a balanced connected subgraph of $G$ having $\min \{b,r\}$ vertices of each color.
	\end{theorem}

	\subsection{Bipartite graphs, properly colored}  
	In this section, we describe a polynomial-time algorithm for the  \bcs~problem where the input graph is a bipartite graph whose nodes are colored red/blue according to proper 2-coloring of vertices in a graph. We show that there is a balanced connected subgraph of $G$ having $\min \{b,r\}$ vertices of each color where $G$ contains $r$ red vertices and $b$ blue vertices. Note that we earlier showed that the \bcs~problem is NP-hard in bipartite graphs whose vertices are colored red/blue arbitrarily; here, we insist on the coloring being a proper coloring (the construction in the hardness proof had adjacent pairs of vertices of the same color). 
	We begin with the following lemma.
	
	\begin{lemma}\label{lemma-2colorable}
		Consider a tree $T$ (which is necessarily bipartite) and a proper $2$-coloring of its nodes, with $r$ red nodes and $b$ blue nodes.  If $r<b$, then $T$ has at least one blue leaf.
	\end{lemma}
	\begin{proof}
		We prove it by contradiction. Let there is no blue leaf. Now assign any blue node say $b_{r}$ as a root. Note that it always exists. Now $b_{r}$ is at level $0$ and $b_{r}$ has degree at least $2$. Otherwise, $b_{r}$ is a leaf with blue color. We put all the adjacent vertices of $b_{r}$ in level~$1$. This level consists of only red vertices. In level~$2$ we put all the adjacent vertices of level~$1$. So level~$2$ consists of only blue vertices. This way we traverse all the vertices in $T$ and let that we stop at $k^{th}$-level. $k$ cannot be even as all the vertices in even level are blue. So $k$ must be odd. Now for each $0 \leqslant i \leqslant \frac{k-1}{2}$, in the vertices of (level $2i$ $\cup$ level $(2i+1)$), number of blue vertices is at most the number of red vertices. Which leads to the contradiction that $r<b$. Hence there exists at least one leaf with blue color. \qed
	\end{proof}
	
	We are now ready to describe the algorithm. We first find a spanning tree $T$ in $G$.  If $r=b$ then $T$ itself is a maximum balanced subtree (subgraph also) of $G$. Without loss of generality assume that $r<b$.  So by Lemma~\ref{lemma-2colorable}, $T$ has at least $1$ blue vertex. Now we remove that blue vertex from $T$. Using similar reason, we repetitively remove $(b-r)$ blue vertices from $T$. Finally, $T$ becomes balanced subgraph of $G$, with $r$ many vertices of each color.
	
	\vspace{.2cm}
	
	\noindent {\bf Running time:}
	Finding a spanning tree in $G$ requires $O(n^2)$ time. To find all the leaves in the tree $T$ requires $O(n^2)$ time (breadth first search). Hence the total time is needed is $O(n^2)$.  
	
	Now, we state the following theorem.
	
	\begin{theorem}
		Given a bipartite graph $G$ with a proper 2 coloring ($r$ red or $b$ blue vertices), then in $O(n^2)$ time we can find a balanced connected subgraph in $G$ having $\min \{b,r\}$ vertices of each color.
	\end{theorem}
	
	\subsection{Graphs of diameter~2}
	In this section, we give a polynomial time algorithm which solves the \bcs-problem where the input graph has diameter 2. Let $G(V,E)$ be such a graph which contains $b$ blue vertex set $B$ and $r$ red vertex set $R$. We find a balanced connected subgraph $H$ of $G$ having $\min\{b,r\}$ vertices of each color. Assume that $b<r$. This can be done in two phases. In phase 1, we generate an induced connected subgraph $G'$ of $G$ such that (i) $G'$ contains all the vertices in $B$, and (ii) the number of vertices in $G'$ is at most $(2b-1)$. In phase 2, we find $H$ from $G'$. 
	
	\begin{description}
		\item[Phase 1] To generate $G'$, we use the following result.
		
		\begin{lemma}\label{dia_2}
			Let $G=(V,E)$ be a graph of diameter 2. Then for any pair of non adjacent vertices $u$ and $v$ from $G$, there always exists a vertex $w$ such that both $(u,w)\in E$ and $(v,w) \in E$.
		\end{lemma}

		We first include $B$ in $G'$. Now we have the following two cases.
		\begin{description}
			
			\item \noindent{\textbf{Case 1:}} The induced subgraph $G[B]$ of $B$ is connected. In this case, $G'$ is $G[B]$.
			
			\item \noindent{\textbf{Case 2:}} The induced subgraph $G[B]$ of $B$ is not connected.
			Assume that $G[B]$  has $k(>1)$ components. Let $B_{1},B_{2},..., B_{k}$ be $k$ disjoints sets of vertices such that each induced subgraph $G[B_i]$ of $B_i$ in $G$ is connected. Now using Lemma \ref{dia_2}, any two vertices $v_i \in B_i$ and $v_j \in B_j$ are adjacent to a vertex say $u_\ell \in R$. We repetitively apply  Lemma \ref{dia_2} to merge all the $k$ subgraphs into a larger  graph. We need at most $(k-1)$ red vertices to merge $k$ subgraph. We take this larger graph as the graph $G'$. 
		\end{description}
		
		\vspace{.2cm}
		
		\item[Phase 2] In this phase, we find the balanced connected subgraph $H$ with $b$ vertices of each color. Note that the graph $G'$ generated in phase 1 contains $b$ blue and at most $(b-1)$ red vertices. Assume that $G'$ contains $b'$ red vertices. We add $(b-b')$ red vertices from $G \setminus G'$ to $G'$. This is possible since $G$ in connected.
	\end{description}
	
	\vspace{.2cm}

	\noindent {\bf Running time:} In phase 1, first finding all the blue vertices and it's induced subgraph takes $O(n^2)$ time. Now to merge all the $k$ components into a single component which is $G'$ needs  $O(n^2)$ time.  In phase 2, adding $(b-b')$ red vertices to $G'$ takes $O(n^2)$ time as well. Hence, total time requirement  is $O(n^2)$.  
	
	\begin{theorem}
		Given a graph $G = (V, E)$ of diameter~$2$, where the vertices in $G$ are colored either red or blue. If $G$ has $b$ blue and $r$ red vertices then, in $O(n^2)$ time we can find a balanced connected subgraph in $G$ having $\min \{b,r\}$ vertices of each color.
	\end{theorem}
	
	\section{Conclusions and open questions}
	We have introduced the problem of finding largest size (cardinality of the vertex set) balanced connected subgraph in a simple connected graph. We have seen that this problem is \np-complete for bipartite graphs, chordal graphs, or planar graph. We have given polynomial time algorithms for solving this problem for trees, graphs with proper $2$ coloring, split graphs and graphs with diameter~$2$. So the obvious question is can other special classes of graphs be found to yield polynomial time algorithms? For example, outer planar graphs, interval graphs, regular graphs, permutation graphs etc. Here we give another open question. Let $G$ be a given graph and $ OPT$ be the number of vertices in an optimal solution of \bcs~problem.  Is there any polynomial time $(\alpha,\beta)$  approximation algorithm which yields a solution $ H $ such that minimum number of blue and red vertices in $ H $ is at most $ \alpha \times OPT $ and difference between the number of blue and red vertices in $ H $ is at most $\beta$?


\begin{thebibliography}{10}
	\providecommand{\url}[1]{\texttt{#1}}
	\providecommand{\urlprefix}{URL }
	\providecommand{\doi}[1]{https://doi.org/#1}
	
	\bibitem{abellanas1999bipartite}
	Abellanas, M., Garc{\i}a, J., Hern{\'a}ndez, G., Noy, M., Ramos, P.: Bipartite
	embeddings of trees in the plane. Discrete Applied Mathematics
	\textbf{93}(2-3),  141--148 (1999)
	
	\bibitem{aichholzer2015balanced}
	Aichholzer, O., Atienza, N., Fabila-Monroy, R., Perez-Lantero, P.,
	D{\i}az-B{\'a}{\~n}ez, J.M., Flores-Pe{\~n}aloza, D., Vogtenhuber, B.,
	Urrutia, J.: Balanced islands in two colored point sets in the plane. arXiv
	preprint arXiv:1510.01819  (2015)
	
	\bibitem{Backes2012}
	Backes, C., Rurainski, A., Klau, G.W., M\"{u}ller, O., St\"{o}ckel, D.,
	Gerasch, A., K\"{u}ntzer, J., Maisel, D., Ludwig, N., Hein, M., Keller, A.,
	Burtscher, H., Kaufmann, M., Meese, E., Lenhof, H.P.: An integer linear
	programming approach for finding deregulated subgraphs in regulatory
	networks. Nucleic Acids Research  \textbf{40}(6), ~e43 (2012)
	
	\bibitem{balachandran2017system}
	Balachandran, N., Mathew, R., Mishra, T.K., Pal, S.P.: System of unbiased
	representatives for a collection of bicolorings. arXiv preprint
	arXiv:1704.07716  (2017)
	
	\bibitem{bereg2015balanced}
	Bereg, S., Hurtado, F., Kano, M., Korman, M., Lara, D., Seara, C., Silveira,
	R.I., Urrutia, J., Verbeek, K.: Balanced partitions of 3-colored geometric
	sets in the plane. Discrete Applied Mathematics  \textbf{181},  21--32 (2015)
	
	\bibitem{biniaz2014bottleneck}
	Biniaz, A., Maheshwari, A., Smid, M.H.: Bottleneck bichromatic plane matching
	of points. In: CCCG (2014)
	
	\bibitem{Carvajal2013}
	Carvajal, R., Constantino, M., Goycoolea, M., Vielma, J.P., Weintraub, A.:
	Imposing connectivity constraints in forest planning models. Operations
	Research  \textbf{61}(4),  824--836 (2013)
	
	\bibitem{Chen2017}
	Chen, C.Y., Grauman, K.: Efficient activity detection in untrimmed video with
	max-subgraph search. IEEE Transactions on Pattern Analysis and Machine
	Intelligence  \textbf{39}(5),  908--921 (2017)
	
	\bibitem{Chen2012}
	Chen, C.Y., Grauman, K.: Efficient activity detection with max-subgraph search.
	2012 IEEE Conference on Computer Vision and Pattern Recognition pp.
	1274--1281 (2012)
	
	\bibitem{crowston2013maximum}
	Crowston, R., Gutin, G., Jones, M., Muciaccia, G.: Maximum balanced subgraph
	problem parameterized above lower bound. Theoretical Computer Science
	\textbf{513},  53--64 (2013)
	
	\bibitem{derhy2009finding}
	Derhy, N., Picouleau, C.: Finding induced trees. Discrete Applied Mathematics
	\textbf{157}(17),  3552--3557 (2009)
	
	\bibitem{Dilkina2010}
	Dilkina, B., Gomes, C.P.: Solving connected subgraph problems in wildlife
	conservation. In: Lodi, A., Milano, M., Toth, P. (eds.) Integration of AI and
	OR Techniques in Constraint Programming for Combinatorial Optimization
	Problems. pp. 102--116. Springer Berlin Heidelberg, Berlin, Heidelberg (2010)
	
	\bibitem{Dittrich2008}
	Dittrich, M.T., Klau, G.W., Rosenwald, A., Dandekar, T., M\"{u}ller, T.:
	Identifying functional modules in protein-protein interaction networks: an
	integrated exact approach. Bioinformatics  \textbf{24}(13),  i223--i231
	(2008)
	
	\bibitem{dumitrescu2001matching}
	Dumitrescu, A., Kaye, R.: Matching colored points in the plane: some new
	results. Computational Geometry  \textbf{19}(1),  69--85 (2001)
	
	\bibitem{dumitrescu2002partitioning}
	Dumitrescu, A., Pach, J.: Partitioning colored point sets into monochromatic
	parts. International Journal of Computational Geometry \& Applications
	\textbf{12}(05),  401--412 (2002)
	
	\bibitem{El-Kebir2014}
	El{-}Kebir, M., Klau, G.W.: Solving the maximum-weight connected subgraph
	problem to optimality. CoRR  \textbf{abs/1409.5308} (2014)
	
	\bibitem{feige2001dense}
	Feige, U., Peleg, D., Kortsarz, G.: The dense k-subgraph problem. Algorithmica
	\textbf{29}(3),  410--421 (2001)
	
	\bibitem{Felix1998}
	Felix, L.H., R., D.D.: Decomposition algorithms for the maximum-weight
	connected graph problem. Naval Research Logistics (NRL)  \textbf{45}(8),
	817--837 (1998)
	
	\bibitem{Garey1979}
	Garey, M.R., Johnson, D.S.: Computers and Intractability: {A} Guide to the
	Theory of {NP}-Completeness. W. H. Freeman (1979)
	
	\bibitem{Hochbaum1994}
	Hochbaum, D.S., Pathria, A.: Node-optimal connected $k$-subgraphs (1994)
	
	\bibitem{Johnson1985}
	Johnson, D.S.: The {NP}-completeness column: An ongoing guide. Journal of
	Algorithms  \textbf{6}(1),  145 -- 159 (1985)
	
	\bibitem{kaneko2003discrete}
	Kaneko, A., Kano, M.: Discrete geometry on red and blue points in the plane—a
	survey—. In: Discrete and computational geometry, pp. 551--570. Springer
	(2003)
	
	\bibitem{kaneko2017balancing}
	Kaneko, A., Kano, M., Watanabe, M.: Balancing colored points on a line by
	exchanging intervals. Journal of Information Processing  \textbf{25},
	551--553 (2017)
	
	\bibitem{kierstead1992colorful}
	Kierstead, H.A., Trotter, W.T.: Colorful induced subgraphs. Discrete
	mathematics  \textbf{101}(1-3),  165--169 (1992)
	
	\bibitem{Yamamoto2009}
	Yamamoto, T., Bannai, H., Nagasaki, M., Miyano, S.: Better decomposition
	heuristics for the maximum-weight connected graph problem using betweenness
	centrality. In: Gama, J., Costa, V.S., Jorge, A.M., Brazdil, P.B. (eds.)
	Discovery Science. pp. 465--472. Springer Berlin Heidelberg, Berlin,
	Heidelberg (2009)
	
\end{thebibliography}
\end{document}